\documentclass[10pt]{bmc_article}

\usepackage{hyperref}
\usepackage{cite} 
\usepackage{url}  
\usepackage{ifthen}  
\usepackage{multicol}   
\usepackage[utf8]{inputenc}
\usepackage[english]{babel}
\usepackage{tikz}
\usetikzlibrary{positioning}
\usepackage{subfig}
\usepackage[lowercase]{theoremref}
\usepackage{amsmath, amssymb, amsfonts, amsthm}
\usepackage{graphicx}
\usepackage{nicefrac}

\newboolean{publ}

\newtheorem{theorem}{Theorem}
\newtheorem{lemma}[theorem]{Lemma}

\newcommand{\dcjd}[2]{\mathrm{d}_{dcj}(#1, #2)}
\newcommand{\ms}{\mathrm{ms}}
\newcommand{\ts}{\mathrm{ts}}
\newcommand{\MBG}{\mathrm{BG}}

\newenvironment{bmcformat}{\baselineskip20pt\sloppy\setboolean{publ}{false}}{\baselineskip20pt\sloppy}

\begin{document}
\begin{bmcformat}

\title{On pairwise distances and median score of three genomes under DCJ}

\author{Sergey Aganezov, Jr.$^1$%
       \quad and \quad
         Max A. Alekseyev\correspondingauthor$^{1,2}$%
         \email{Max A. Alekseyev\correspondingauthor - \href{mailto:maxal@cse.sc.edu}{maxal@cse.sc.edu}}%
      }

\address{%
    \iid(1)Algorithmic Biology Laboratory, St. Petersburg Academic University, St. Petersburg, Russia\\
    \iid(2)Department of Computer Science and Engineering, University of South Carolina, Columbia, SC, USA
}%

\maketitle

\vskip-5pt

\begin{abstract}

In comparative genomics, the rearrangement
distance between two genomes (equal the minimal number of genome rearrangements required to transform them into a single genome)
is often used for measuring their evolutionary remoteness.
Generalization of this measure to three genomes is known as the \emph{median score} (while a resulting genome is called \emph{median genome}).
In contrast to the rearrangement distance between two genomes which can be computed in linear time, 
computing the median score for three genomes is NP-hard. 
This inspires a quest for simpler and faster approximations for the median score, the most natural of which appears to be the halved sum 
of pairwise distances which in fact represents a lower bound for the median score.

In this work, we study relationship and interplay of pairwise distances between three genomes and their median score under 
the model of Double-Cut-and-Join (DCJ) rearrangements.
Most remarkably we show that while a rearrangement may change the sum of pairwise distances by at most 2 
(and thus change the lower bound by at most 1), 
even the most ``powerful'' rearrangements in this respect that increase the lower bound by 1 
(by moving one genome farther away from each of the other two genomes), 
which we call \emph{strong}, do not necessarily affect the median score. 
This observation implies that the two measures are not as well-correlated as one's intuition may suggest.

We further prove that the median score attains the lower bound exactly on the triples of genomes 
that can be obtained from a single genome with strong rearrangements.
While the sum of pairwise distances with the factor $\nicefrac{2}{3}$ represents an upper bound for the median score, its tightness remains unclear.
Nonetheless, we show that 
the difference of the median score and its lower bound is not bounded by a constant.
\end{abstract}

\ifthenelse{\boolean{publ}}{\begin{multicols}{2}}{}

\section*{Background}

The number of large-scale rearrangements (such as reversals, translocations, fissions, and fusions) between two genomes 
is often used as a measure of their evolutionary remoteness.
The minimal number of such rearrangements required to transform one genome into the other is called \emph{rearrangement distance}. 
Computing rearrangement distances between the genomes of interest is often a pre-requisite for their comparative analysis (e.g., phylogeny reconstruction).

\emph{Double-Cut-and-Join (DCJ) rearrangements}~\cite{10,11} (also known as \emph{2-breaks}~\cite{2}) represent a convenient 
model of reversals, translocations, fissions, and fusions, which allows one to compute the corresponding \emph{DCJ distance} between two genomes in linear time.

Phylogeny reconstruction for three given genomes involves reconstruction of their \emph{median genome} that minimizes the total distance from the given genomes. This minimal
total distance, called the \emph{median score}~\cite{6}, represents a natural generalization of the DCJ distance to the case of three genomes.
In contrast to DCJ distance between two genomes, computing the median score of three genomes is NP-hard~\cite{5, 6}. 
While there exist exact~\cite{8,9} and heuristic~\cite{7, 12, 13} algorithms for this problem, they can hardly be used for large genomes.
This inspires a quest for simpler and faster approximations for the median score.

The simplest and easily computable approximation for the median score of three genomes is given by the sum of their pairwise DCJ distances, 
which we call the \emph{triangle score}.
In this work, we study the tightness of this approximation.
In particular, we show that with the factor $\nicefrac{1}{2}$ it represents a lower bound and with the factor $\nicefrac{2}{3}$ it represents an upper bound for the median score.
We further prove that the lower bound is attained exactly for the triples that can be obtained from a single genome by \emph{strong} rearrangements that increase 
the triangle score by 2 (by moving one genome farther away from each of the other two genomes). In other words, strong rearrangements are those that increase the lower bound by 1. 
From this perspective, it natural to expect that strong rearrangements always increase the median score. However, we disprove this expectation with a counterexample.

While tightness of the upper bound remains unclear, we remark that a better upper bound for the median score may improve performance of algorithms for computing median score 
based on the adequate subgraph decomposition~\cite{8,9}.
Still, we make an initial step in this direction by proving that there is no upper bound equal the lower bound plus a constant.

\section*{Methods}
\subsection*{Breakpoint graphs and genome rearrangements}

In this work, we focus on circular genomes consisting of one or more circular chromosomes. 
A circular genome on a set of $n$ genes (say, $\{1,2,\dots, n\}$) can be represented as a perfect matching on $2n$ vertices~\cite{1, 4} where
each gene is represented with a pair of vertices, corresponding to the gene's extremities: ``head'' and ``tail''; while
each adjacency between two genes in the genome is represented with an edge between respective extremities. 
\emph{Breakpoint graph} of genomes $A_1, A_2, \dots, A_k$, denoted $\MBG(A_1,\dots,A_k)$ 
is defined as the superposition of $k$ perfect matchings
representing given genomes, each of its own color~\cite{4}. We refer to edges representing adjacencies in the genome $A_i$ as \emph{$A_i$-edges} ($i=1,2,\dots,k$).
When all genomes $A_1, A_2, \dots, A_k$ are identical, their breakpoint graph is called an \emph{identity breakpoint graph}. 
Every identity breakpoint graph consists of \emph{trivial multicycles} formed by $k$ parallel edges of all colors.

A \emph{DCJ} rearrangement in genome $A$ replaces a pair of $A$-edges with another pair of $A$-edges forming matching on the same four vertices.
The \emph{DCJ distance} between genomes $A$ and $B$ on the same set of $n$ genes, denoted with $\dcjd{A}{B}$, is defined as 
the minimal number of DCJs required to transform one genome into the other.
The DCJ distance $\dcjd{A}{B}$ is closely connected with the number $c(A,B)$ of alternating cycles (i.e., cycles with edges of alternating colors) 
in the breakpoint graph $\MBG(A,B)$ by the formula: $\dcjd{A}{B} = n - c(A,B)$. We remark that $c(A,B)$ may range from $1$ to $n$, implying that $0\leq \dcjd{A}{B}\leq n-1$.
A single DCJ in $A$ or $B$ can change $\dcjd{A}{B}$ by at most $1$~\cite{2, 3, 10}.

\begin{figure*}[!t]
\centering 
\includegraphics{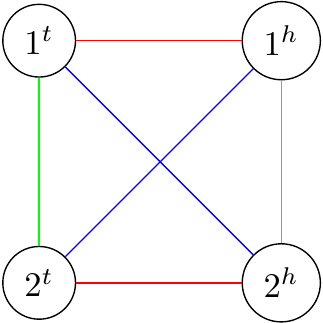}
\caption{Breakpoint graph $\MBG(A,B,C)$ of genomes $A = (1) (2)$ (red edges), $B = (1, 2)$ (blue edges), and $C = (1, -2)$ (green edges),
where $(1^t - 1^h - 2^t - 2^h - 1^t)$ is an $AB$-cycle, $(1^t - 1^h - 2^h - 2^t - 1^t)$ is an $AC$-cycle, and $(1^t - 2^h - 2^t - 1^h - 1^t)$ is a $BC$-cycle.}
\label{fig:MBG_example}
\end{figure*}

Let $A,B,C$ be genomes on a set of $n$ genes.
Their breakpoint graph $\MBG(A,B,C)$ is formed by $A$-edges, $B$-edges, and $C$-edges   
so that each pair of genomes define alternating cycles, 
called respectively $AB$-cycles, $AC$-cycles, and $BC$-cycles (Fig.~\ref{fig:MBG_example}). 
We further define the \emph{triangle score} $\ts(A,B,C)$ as the sum of pairwise DCJ distances:
$$\ts(A,B,C) = \dcjd{A}{B} + \dcjd{A}{C} + \dcjd{B}{C}.$$
Since a DCJ in one of the genomes can change each of the two corresponding distances by at most $1$, 
it can change $\ts(A,B,C)$ by at most 2.

\begin{lemma}\label{about_two_a_edges}
Let $A,B,C$ be genomes on the same set of $n$ genes such that 
$\dcjd{A}{B}$ and $\dcjd{A}{C}$ are less than $n-1$. 
Then in $\MBG(A,B,C)$ there exists a pair of $A$-edges that belong to two distinct $AB$-cycles and two distinct $AC$-cycles.
\end{lemma}

\begin{proof}
Since $\dcjd{A}{B}<n-1$, there exist at least two distinct $AB$-cycles in the breakpoint graph $\MBG(A,B,C)$.
Therefore, the $AB$-cycles define a partition of the set of $A$-edges $S_A$ into two or more nonempty subsets: $S_A = P_1 \cup P_2 \cup \dots$

Similarly, since $\dcjd{A}{C}<n-1$, there exist at least two distinct $AC$-cycles in $\MBG(A,B,C)$ so that the $AC$-cycles
define a partition of $S_A$ into two or more nonempty subsets: $S_A = Q_1 \cup Q_2 \cup \dots$.
Intersecting the subsets in the two partitions, we get a partition of $S_A$ into subsets, each consisting of $A$-edges 
that belong to the same $AB$-cycle and the same $AC$-cycle: $S_A = \bigcup_{i,j} (P_i \cap Q_j)$.

Suppose that there is no required pair of $A$-edges, implying that for any two non-empty intersections $P_i \cap Q_j$ and $P_{i'}\cap Q_{j'}$, 
we have either $i=i'$ or $j=j'$.
Without loss of generality, assume that $P_1 \cap Q_1$ is non-empty. Then for every $i>1$ and $j>1$, $P_i \cap Q_j$ must be empty, implying that $P_i\subset Q_1$.
In particular, $P_2\cap Q_1=P_2$ is non-empty and by the same reasoning, we have $P_1 \subset Q_1$. Therefore, $P_i\subset Q_1$ for all $i$, implying that $Q_1 = S_A$, 
a contradiction to non-emptiness of $Q_2$. This contradiction proves that a required pair of $A$-edges exists.
\end{proof}

\begin{figure*}[!t]
\centering 
\includegraphics{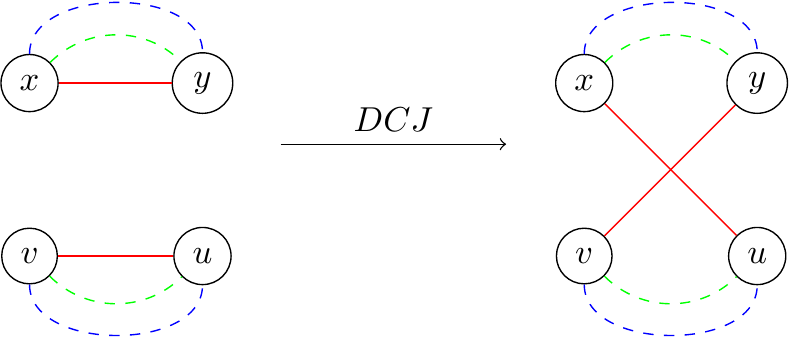}
\caption{$A$-edges (red) that belong to distinct $AB$-cycles and distinct $AC$-cycles, denoted by dashed blue lines and dashed green lines, respectively (left panel).
A DCJ on these $A$-edges merges these $AB$-cycles and $AC$-cycles into a single $AB$-cycle and a single $AC$-cycle, and thus increases $\ts(A,B,C)$ by 2 (right panel).}
\label{fig:dcj_d_ABC_2}
\end{figure*}

\begin{theorem}\label{strong_DCJ}
If between three genomes $A,B,C$ on the same set of $n$ genes at least two pairwise DCJ distances are less than $n-1$, 
then there exists a DCJ (called \emph{strong}) that increases $\ts(A,B,C)$ by 2. 
\end{theorem}

\begin{proof}
Without loss of generality, we assume that $\dcjd{A}{B}<n-1$ and $\dcjd{A}{C}<n-1$. 
By Lemma~\ref{about_two_a_edges}, there are $A$-edges $(x,y)$ and $(u,v)$ that belong to distinct $AB$-cycles and distinct $AC$-cycles in $\MBG(A,B,C)$.
Using any DCJ on these edges (Fig.~\ref{fig:dcj_d_ABC_2}), we decrease the number of $AB$-cycles as well as the number of $AC$-cycles by $1$. 
Since $BC$-cycles remain intact, this DCJ increases $\ts(A,B,C)$ by 2. 
\end{proof}

\begin{theorem}\label{ThPDachieve}
Let $(p,q,r)$ be a triple of integers from the interval $[0,n-1]$, satisfying the triangle inequality.
There exist three genomes on a set of $n$ genes whose pairwise DCJ distances are $(p,q,r)$.
Moreover, these genomes can be obtained with $\left\lfloor \tfrac{p+q+r}{2} \right\rfloor$ strong DCJs and possibly one other DCJ
(when $p+q+r$ is odd) from a single genome.
\end{theorem}

\begin{proof}
Without loss of generality, we assume that $p\leq q\leq r$ and that $p=\dcjd{A}{B}$, $q=\dcjd{A}{C}$, and $r=\dcjd{B}{C}$ where $A,B,C$ are the genomes being constructed.

If $p+q+r$ is even, we start with $A$, $B$, $C$ being the same genome and notice that
$$(p,q,r) = \frac{p+q-r}{2}\cdot (1,1,0) + \frac{p+r-q}{2}\cdot (1,0,1) + \frac{q+r-p}{2}\cdot (0,1,1)$$
where by the triangle inequality all coefficients are nonnegative.
This identity instructs us to apply $\tfrac{p+q-r}{2}$ strong DCJs to the genome $A$ (increasing both $\dcjd{A}{B}$ and $\dcjd{A}{C}$), 
$\tfrac{p+r-q}{2}$ strong DCJs to genome $B$, and $\tfrac{q+r-p}{2}$ strong DCJs to genome $C$. 
Existence of such strong DCJs is guaranteed by Theorem~\ref{strong_DCJ}.

If $p+q+r$ is odd, then we have $p>0$ as otherwise the triangle inequality would imply $q=r$ and thus even $p+q+r$.
In this case we start with $A$, $B$, $C$ being three genomes with pairwise DCJ distances $(1,1,1)$ such that $\MBG(A,B,C)$
consists of trivial multicycles, except for the vertices $1^t$, $1^h$, $2^t$, and $2^h$ connected as in Fig.~\ref{fig:MBG_example}.
It is easy to see that these genomes can be obtained from the same genome by one strong DCJ and one non-strong DCJ.
We further increase the pairwise DCJ distances between genomes $A$, $B$, $C$ by $(p',q',r')=(p-1,q-1,r-1)$ 
with $\tfrac{p'+q'+r'}{2}$ strong DCJs as above (notice that $p'+q'+r'$ is even and $(p',q',r')$ satisfies the triangle inequality).
Therefore, the total number of strong DCJs in this case is 
$$1 + \frac{p'+q'+r'}{2} = \frac{p+q+r-1}{2} = \left\lfloor \frac{p+q+r}{2} \right\rfloor.$$
\end{proof}

While all triples of pairwise DCJ distances are achievable with strong DCJs, not 
all breakpoint graphs of three genomes can be constructed from an identity breakpoint graph this way. 
In particular, Figure~\ref{fig:non_coll_MGB} gives an example of breakpoint graph $\MBG(A,B,C)$ such that $\ts(A,B,C)$ cannot be decreased by 2 with a DCJ.
In this example, we have $\ts(A,B,C) = 6$ but there is no sequence of three DCJs that would produce $\MBG(A,B,C)$ from an identity breakpoint graph.

\begin{figure*}[!t]
\centering 
\includegraphics{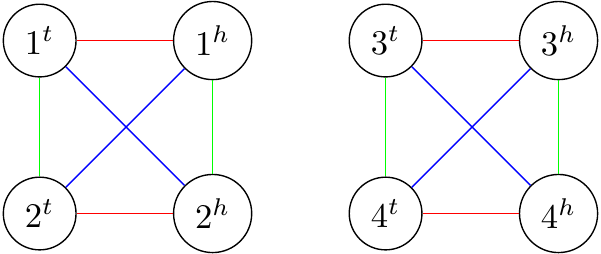}
\caption{Breakpoint graph $\MBG(A,B,C)$ of genomes $A = (1) (2) (3) (4)$ (red edges), $B = (1, 2) (3,4)$ (blue edges), and $C = (1, -2) (3, -4)$ (green edges) with the property that
no DCJ can decrease $\ts(A,B,C)$ by 2.}
\label{fig:non_coll_MGB}
\end{figure*}

In the next section we demonstrate how DCJs on three genomes can affect their median score.

\subsection*{Strong rearrangements and median score}

\emph{Median genome problem} for given genomes $A,B,C$ is to find a genome (which is called \emph{median genome} and may not be unique)
that attains the \emph{median score}~\cite{6}:
$$\ms(A,B,C) = \min_{M} \dcjd{A}{M} + \dcjd{B}{M} + \dcjd{C}{M}.$$
The median problem can be alternatively posed as finding the minimal number (equal $\ms(A,B,C)$) of DCJs required to transform the genomes $A,B,C$ 
into a single (median) genome (or, vice versa, to obtain $A,B,C$ from a single genome). 
In fact, this formulation further generalizes and becomes particularly useful for phylogeny reconstruction of a larger number of genomes~\cite{4}.

From perspective of this formulation, it becomes important to realize what triples of genomes can be obtained from a single genome with strong DCJs.
We start with proving a helpful lemma and bounds on the median score in terms of the triangle score.

\begin{lemma}\label{LemmaMS}
For three genomes on the same set of genes, a DCJ in one of the genomes may change their median score by at most 1.
\end{lemma}

\begin{proof} Let $A,B,C$ be genomes of the same set of genes. Consider a DCJ in a genome $A$ and denote the resulting genome by $A'$.
Let $M$ be a median genome of the genomes $A,B,C$ so that
$$\ms(A,B,C) = \dcjd{A}{M} + \dcjd{B}{M} + \dcjd{C}{M}.$$
Clearly, $\dcjd{A}{M} \geq \dcjd{A'}{M} - 1$ and thus
\begin{multline*}
\ms(A,B,C) = \dcjd{A}{M} + \dcjd{B}{M} + \dcjd{C}{M} \\
\geq \dcjd{A'}{M} + \dcjd{B}{M} + \dcjd{C}{M} - 1 \geq \ms(A',B,C) - 1,
\end{multline*}
i.e., $\ms(A',B,C) - \ms(A,B,C) \leq 1$.
Symmetrically, we also have $\ms(A,B,C) - \ms(A',B,C) \leq 1$ and thus $|\ms(A,B,C) - \ms(A',B,C)| \leq 1$.
\end{proof}

\begin{theorem}\label{ThBounds}
For genomes $A,B,C$ on the same set of genes, we have
$$\frac{1}{2}\cdot\ts(A,B,C) \leq \ms(A,B,C) \leq \frac{2}{3}\cdot\ts(A,B,C).$$
\end{theorem}

\begin{proof} 
Consider a transformation of each of the genomes $A,B,C$ into a median genome with DCJs.
The total number of DCJs in this transformation is $\ms(A,B,C)$. Since each DCJ decreases $\ts(A,B,C)$ by at most 2, we have 
$\ts(A,B,C) \leq 2\cdot \ms(A,B,C)$, implying that $\nicefrac{1}{2}\cdot\ts(A,B,C) \leq \ms(A,B,C)$.

On the other hand, the number of DCJs in any transformation of the genomes $A,B,C$ into the genome $A$ is at least $\ms(A,B,C)$, implying that
$\dcjd{B}{A} + \dcjd{C}{A} \geq \ms(A,B,C)$. Similarly, we have $\dcjd{A}{B} + \dcjd{C}{B} \geq \ms(A,B,C)$ and $\dcjd{A}{C} + \dcjd{B}{C} \geq \ms(A,B,C)$. 
Summing up these three inequalities, we get $2\cdot \ts(A,B,C) \geq 3\cdot \ms(A,B,C)$ and thus $\ms(A,B,C)\leq \nicefrac{2}{3}\cdot\ts(A,B,C)$.
\end{proof}

We remark that the lower bound and a slightly better upper bound 
$\ms(A,B,C)\leq \min\{ \dcjd{A}{B} + \dcjd{A}{C}, \dcjd{A}{B} + \dcjd{B}{C}, \dcjd{A}{C} + \dcjd{B}{C}\}$ was also used in~\cite{9}.

The following theorem classifies all triples of genomes for which the median score coincides with its lower bound and links them with the genomes
constructed in Theorem~\ref{ThPDachieve}.

\begin{theorem}\label{ThClassify}
For genomes $A,B,C$ on the same set of genes, $\ms(A,B,C) = \nicefrac{1}{2}\cdot\ts(A,B,C)$ 
if and only if $A,B,C$ can be obtained from a single genome with strong DCJs.
\end{theorem}

\begin{proof} Suppose that $\ms(A,B,C) = \nicefrac{1}{2}\cdot\ts(A,B,C)$. 
Let $M$ be a median genome of the genomes $A,B,C$. Then $A,B,C$ can be obtained from $M$ with $\ms(A,B,C) = \nicefrac{1}{2}\cdot\ts(A,B,C)$ DCJs. 
This transformation increases the triangle score from $\ts(M,M,M)=0$ to $\ts(A,B,C)$.
Since each of $\nicefrac{1}{2}\cdot\ts(A,B,C)$ DCJs can increase the triangle score by at most 2, they all must be strong.

Vice versa, suppose that $A,B,C$ are obtained from a single genome with strong DCJs.
Lemma~\ref{LemmaMS} implies that a strong DCJ does not increase the difference between the median score and its lower bound.
Since the transformation starts with the median score equal its lower bound (i.e., their difference is 0), 
they further remain equal along the whole transformation, resulting in $\ms(A,B,C) = \nicefrac{1}{2}\cdot\ts(A,B,C)$.
\end{proof}

It remains unclear how tight is the upper bound given in Theorem~\ref{ThBounds}, 
while a better upper bound may improve performance of algorithms for computing median score based on the adequate subgraph decomposition~\cite{8,9}.
Below we prove however that the upper bound cannot be equal to the lower bound plus a constant.

\begin{theorem}\label{ThLBplusC}
The difference $\ms(A,B,C) - \nicefrac{1}{2}\cdot\ts(A,B,C)$ of the median score and its lower bound is not bounded from above by a constant.
\end{theorem}

\begin{proof} To prove the theorem, for every $n=1,2,\dots$, we will construct three genomes $A_n, B_n, C_n$ on the same $4n$ genes for which 
$\ms(A_n,B_n,C_n) - \nicefrac{1}{2}\cdot\ts(A_n,B_n,C_n) = n$.

We start with genomes $A_1 = (1)\ (2)\ (3)\ (4)$, $B_1 = (1, 2)\ (3, 4)$, and $C_1 = (1, -2)\ (3, -4)$.
The breakpoint graph $\MBG(A_1,B_1,C_1)$ consists of two strongly adequate subgraphs (Fig.~\ref{fig:non_coll_MGB}).
We have $\ms(A_1,B_1,C_1) = 4$ and $\ts(A_1,B_1,C_1)=6$, resulting in $\ms(A_1,B_1,C_1) - \nicefrac{1}{2}\cdot\ts(A_1,B_1,C_1) = 1$.

To construct $\MBG(A_n,B_n,C_n)$ we take $n$ copies of $\MBG(A_1,B_1,C_1)$ and relabel their vertices appropriately. In particular, for $n=2$ we get genomes
$A_2 = (1)\ (2)\ (3)\ (4)\ (5)\ (6)\ (7)\ (8)$, $B_2 = (1, 2)\ (3, 4)\ (5, 6)\ (7, 8)$, and $C_2 = (1, -2)\ (3, -4)\ (5, -6)\ (7, -8)$.
Since edges of a median genome do not connect strongly adequate subgraphs of the breakpoint graph~\cite{8,9},
every copy of $\MBG(A_1,B_1,C_1)$ in $\MBG(A_n,B_n,C_n)$ contributes $4$ to the median score.
It is also clear that every copy of $\MBG(A_1,B_1,C_1)$ contributes $6$ to the triangle score, implying that
$\ms(A_n,B_n,C_n) - \nicefrac{1}{2}\cdot\ts(A_n,B_n,C_n) = 4n - 3n = n$.
\end{proof}

We conclude our analysis with the last but not the least observation about the lower bound $\nicefrac{1}{2}\cdot\ts(A,B,C)\leq\ms(A,B,C)$.
According to Lemma~\ref{LemmaMS}, a DCJ in one of the genomes $A,B,C$ can either increase/decrease the right hand side of this inequality (i.e., the median score)
by 1, or keep it intact.
For a strong DCJ (moving one genome farther away from each of the other two genomes), the left hand side of the inequality is increased by 1.
From this perspective, it is very natural to expect that a strong DCJ should also increase the median score (e.g., it was so in the proof of Theorem~\ref{ThClassify}).
Surprisingly, this intuition fails: Figure~\ref{fig:Example} gives a counterexample of a breakpoint graph of three genomes 
with a strong DCJ that does not increase the median score.

\begin{figure*}[!t]
\centering 
\includegraphics{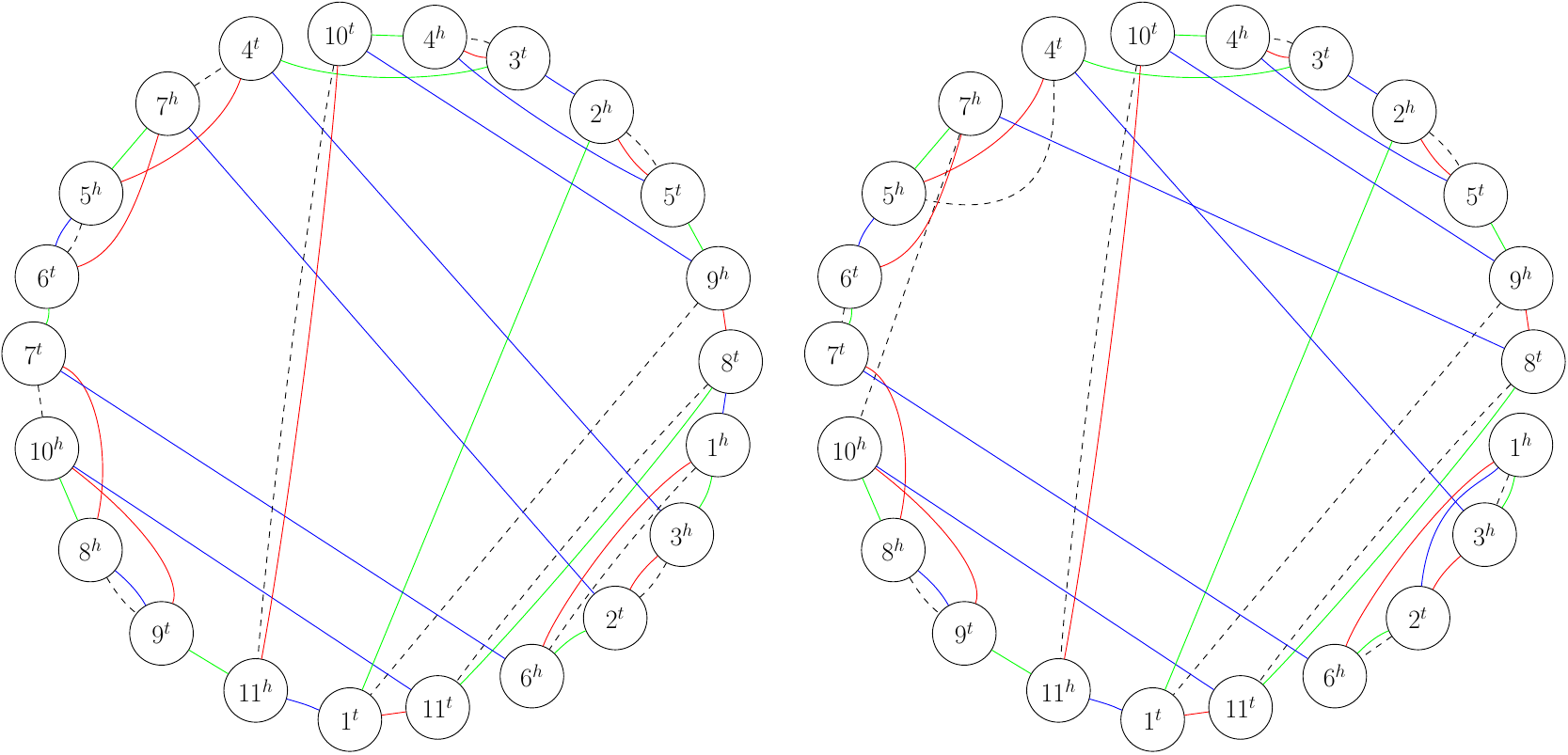}
\caption{\textbf{Left panel:} Breakpoint graph of genomes $A = (1, -6, -7, -8, -9, -10, -11) (2, 5, 4, 3)$ (red edges),
$B = (1, 8, 9, 10, 11) (2, 3, 4, 5, 6, 7)$ (blue edges), $C = (1, -3, 4, 10, -8, 11, 9, 5, -7, 6, 2)$ (green edges), and their median genome 
$M = (1, -6, -5, -2, -3, -4, -7, -10, -11, 8, 9)$ (dashed edges) with $\ts(A,B,C) = 24$ and $\ms(A,B,C)=15$.
The pairwise DCJ distances are $\dcjd{A}{B} = \dcjd{A}{C} = \dcjd{C}{B} = 8$, $\dcjd{A}{M} = 3$, $\dcjd{B}{M} = 5$, and $\dcjd{C}{M} = 7$.\\
\textbf{Right panel:} Breakpoint graph of the same genomes $A$ (red edges), $C$ (green edges), and genome $B' = (1, 2, 3, 4, 5, 6, 7, 8, 9, 10, 11)$ (blue edges) 
obtained from $B$ by a single fusion. The genomes $A$, $B'$, $C$ have a different median genome $M' = (1, -3, -4, -5, -2, -6, 7, -10, -11, 8, 9)$ (dashed edges) 
with the same median score $\ms(A,B',C)=15$ and larger triangle score $\ts(A,B',C) = 26$. The pairwise DCJ distances are 
$\dcjd{A}{B'} = \dcjd{C}{B'} = 9$, $\dcjd{A}{C} = 8$, $\dcjd{A}{M'} = 4$, $\dcjd{B'}{M'} = 6$, and $\dcjd{C}{M'} = 5$. \\
The median genomes $M$ and $M'$ were computed with GASTS~\cite{Xu2011}.
}
\label{fig:Example}
\end{figure*}

\section*{Results and discussion}

We studied two measures of evolutionary remoteness of three genomes $A,B,C$: the triangle score $\ts(A,B,C)$ 
(equal the sum of the pairwise rearrangement distances) 
and the median score $\ms(A,B,C)$ (equal the minimum total rearrangement distance from a single genome).
While computing $\ts(A,B,C)$ takes linear time and computing $\ms(A,B,C)$ is NP-hard, they are connected by the inequality
$\nicefrac{1}{2}\cdot\ts(A,B,C) \leq \ms(A,B,C) \leq \nicefrac{2}{3}\cdot\ts(A,B,C)$ (Theorem~\ref{ThBounds}) giving 
the lower and upper bounds for the median score in terms of the triangle score.

In view of the median genome problem as finding a transformation of the given genomes into a single genome 
(or a reverse transformation of a single genome into the given genomes)
with the smallest number of genome rearrangements, it is important to understand how rearrangements can change the median score and its bounds. 
When $A,B,C$ equal the same genome $M$, the median score trivially coincides with its lower and upper bounds as $\ts(M,M,M) = \ms(M,M,M) = 0$.
Since each rearrangement may change the triangle score by at most $2$ and the median score by at most $1$ (Lemma~\ref{LemmaMS}),
we are particularly interested in strong rearrangements which increase the triangle score by $2$ (and thus increase the lower bound by $1$).

We showed that the median score attains its lower bound (i.e., $\ms(A,B,C) = \nicefrac{1}{2}\cdot\ts(A,B,C)$) exactly on the triples of genomes 
that can be obtained from a single genome with strong rearrangements (Theorem~\ref{ThClassify}). 
We proved that strong rearrangements are common enough to exist for any triple of genomes 
as soon as at least two of their pairwise distances are smaller than the maximum (Theorem~\ref{strong_DCJ}) 
and to produce a triple of genomes with the prescribed pairwise distances (Theorem~\ref{ThPDachieve}). From this perspective, it comes as a total surprise
that strong rearrangements are not ``powerful'' enough to always increase the median score as illustrated by the counterexample in Fig.~\ref{fig:Example}.
This counterexample implies that the median score and the triangle score are not as well-correlated as one's intuition may suggest.

It remains unclear how tight is the upper bound for the median score, while a better upper bound may improve performance of existing algorithms for computing the median score.
Nonetheless, we made an initial step in this direction by proving that there is no upper bound equal the lower bound plus a constant (Theorem~\ref{ThLBplusC}).

\vspace{-.27cm}

\section*{Competing interests}
\ifthenelse{\boolean{publ}}{\small}{}
The authors declare that they have no competing interests.

\vspace{-.27cm}

\section*{Acknowledgements}
\ifthenelse{\boolean{publ}}{\small}{}
This work was partially supported by the Government of the Russian Federation (grant 11.G34.31.0018).

{\ifthenelse{\boolean{publ}}{\footnotesize}{\small}
 \bibliographystyle{bmc_article}  
  \bibliography{article} }     

\ifthenelse{\boolean{publ}}{\end{multicols}}{}

\end{bmcformat}
\end{document}